\documentclass[11pt]{article}

\usepackage{cite}
\usepackage{graphicx}
\usepackage{amsmath}
\usepackage{amsthm}

\usepackage{amsfonts}
\usepackage{amssymb}
\usepackage{fullpage}
\usepackage{geometry}
\usepackage{color}
\usepackage{latexsym}
\usepackage{algorithm}
\usepackage[noend]{algorithmic}
\usepackage{multirow}

\newtheorem{theorem}{Theorem}[section]

\newtheorem{lemma}[theorem]{Lemma}

\theoremstyle{definition}

\renewcommand{\O}{{{\cal{O}}}}
\renewcommand{\S}{{{\cal{S}}}}

\title{An Approximation Algorithm for the Euclidean Bottleneck Steiner Tree Problem}
\date{}
\author{ A. Karim Abu-Affash\thanks{Department of Computer Science, Ben-Gurion University, Beer-Sheva 84105, Israel, {\tt abuaffas@cs.bgu.ac.il}.}}

\begin{document}
\maketitle

\begin{abstract}
Given two sets of points in the plane, $P$ of $n$ terminals and $S$ of $m$ Steiner points, a Steiner tree of $P$ is a tree spanning all points of $P$ and some (or none or all) points of $S$. A Steiner tree with length of longest edge minimized is called a bottleneck Steiner tree. In this paper, we study the Euclidean bottleneck Steiner tree problem: given two sets, $P$ and $S$, and a positive integer $k \le m$, find a bottleneck Steiner tree of $P$ with at most $k$ Steiner points. The problem has application in the design of  wireless communication networks. 

We first show that the problem is NP-hard and cannot be approximated within factor $\sqrt{2}$, unless $P=NP$. Then, we present a polynomial-time approximation algorithm with performance ratio 2.
\end{abstract}
\section{Introduction}
Consider a wireless communication network with $n$ stations, each station has a limited power so that it can only communicate with stations within a limited range, and suppose that, in order to make the network connected and due to budget limits, we are only allowed to put at most $k$ new stations in given potential locations in this network. Clearly, we would like to select locations such that distance between stations as small as possible. This application motivates the following problem: 

{\bf The Bottleneck Steiner Tree ($k$-BST) problem.} Given two sets in the plane, $P$ of terminal points and $S$ of Steiner points, and a positive integer $k$, one is asked to find Steiner tree $T$ of $P$ with at most $k$ Steiner points, such that the bottleneck (i.e., length of the longest edge) of $T$ is minimized. 

In the classical {\em Steiner tree} (ST) problem, the goal is to find a Steiner tree $T$ such the total length of edges of $T$ is minimized. This problem has been shown to be NP-complete~\cite{Garey77} and many approximation algorithms have been proposed~\cite{Berman94, Borchers97, Hougardy99}. A general version of the $k$-BST, where $k=|S|$, has been studied by Sarrafzadeh and Wong~\cite{Sarrafzadeh92}. They showed that this version can be solved in polynomial time. 

Another version, where $S$ is the whole plane $\mathbb{R}^2$, has been studied extensively in the last decade. In~\cite{Wang02}, this version was shown to be NP-hard to approximate within ratio $\sqrt{2}$. The best known upper bound on approximation ratio is $1.866$~\cite{Wang02+}. Bae et al.~\cite{Bae10} presented an $\O(n\log{n})$ time algorithm to the problem for $k=1$ and an $\O(n^2)$ time algorithm for $k=2$. Li et al.~\cite{Li04} presented a ($\sqrt{2}+\epsilon$)-approximation algorithm with inapproximability within $\sqrt{2}$ for a special case of the problem where there should be no edge connecting any two Steiner points in the optimal solution. These versions have many important applications in VLSI design, network communication and computational biology~\cite{Cheng01, Du00, Hwang92, Kahng95}.

We are not aware of any previous work studying our version. However, in this paper, we show that the $k$-BST problem is NP-hard and we present a polynomial-time algorithm with constant factor approximation ratio for the problem.



\section{Hardness Result}\label{sec:Sec2}
Given a set $P$ of $n$ terminals in the plane, a set $S$ of $m$ Steiner points and an integer $k \le m$, the goal in the $k$-BST problem is to find a Steiner tree with at most $k$ Steiner points from $\S$ and bottleneck as small as possible. In this section we prove hardness of the problem.

\begin{theorem}\label{thm:NP-thm}
The $k$-BST problem cannot be approximated within $\sqrt{2}$ in polynomial time, unless $P=NP$.
\end{theorem}

The proof directly follows by a slight modification of the proof of Theorem~1 in~\cite{Wang02}. 


\section{2-Approximation Algorithm}\label{sec:Sec3}

In this section, we develop a polynomial-time approximation algorithm for computing a Steiner tree with at most $k$ Steiner points ($k$-ST for short) such that its bottleneck is at most 2 times the bottleneck of an optimal (minimum-bottleneck) $k$-ST.

Let $G=(V,E)$ be the complete graph over $V=P\cup S$. We assume, without loss of generality, that $E=\{e_1,e_2,\ldots,e_l\}$ such that $|e_1|\le|e_2|\le\ldots\le|e_l|$. It is not hard to see that the bottleneck of an optimal $k$-ST is a length of an edge from $E$. For an edge $e_i\in E$, let $G_i=(V,E_i)$ be the graph with $E_i=\{e_j\in E:|e_j|\le|e_i|\}$. The idea behind our algorithm is to devise a procedure that, for a given edge $e_i\in E$, does one of the following:
\begin{enumerate}
\item[(i)]{It constructs a $k$-ST of $P$ in $G$ with bottleneck at most 2 times $|e_i|$.}
\item[(ii)]{It returns the information that $G_i$ does not contain any $k$-ST of $P$.}
\end{enumerate}

For two points $p,q\in P$, let $\delta_i(p,q)$ be a shortest Steiner path between $p$ and $q$ in $G_i$, i.e., a path connecting $p$ and $q$ with minimum number of Steiner points in $G_i$. Let $G_P=(P,E_P)$ be the complete graph over $P$. For each edge $(p,q)$ in $E_P$, we assign a weight $w(p,q)$ equal to the number of Steiner points in $\delta_i(p,q)$. Let $T$ be a minimum spanning tree of $G_P$ under $w$. We define the normalized weight of $T$ as $C(T)=\sum_{e\in T}\left\lfloor w(e)/2 \right\rfloor$.  

\begin{lemma} \label{lemma:lemma3.3}
If $G_i$ contains a $k$-ST of $P$, then $C(T) \le k$.
\end{lemma}

\begin{proof}[\bf{\textit{Proof:}}]
Let $T^*$ be a $k$-ST of $P$ in $G_i$. A Steiner tree is full if all terminals are leaves. We decompose $T^*$ into a union of full trees. For each full tree $T_j^*$ of $T^*$, we will construct a spanning tree $T'_j$ of the terminals of $T_j^*$ in $G_P$, such that the union of these tree is a spanning tree $T'$ of $P$ in $G_P$ with $C(T')\le k$. We arbitrary select a Steiner point as the root of $T_j^*$; see Figure~\ref{fig:fig2}(a). The construction of $T'_j$ is bottom-up by an iterative process. In each iteration, we select the deepest leaf $p$ in the rooted tree, which is a terminal, and we connect it to its nearest terminal $q$ by an edge of weight equal to the number of Steiner points between them. Let $s$ be the first common parent of $p$ and $q$. We then remove the Steiner points between $p$ and $s$ (in the last iteration, we may remove all of the remaining points). 

\begin{figure}[htp]
    \centering
        \includegraphics[width=.87\textwidth, height=0.3\textwidth]{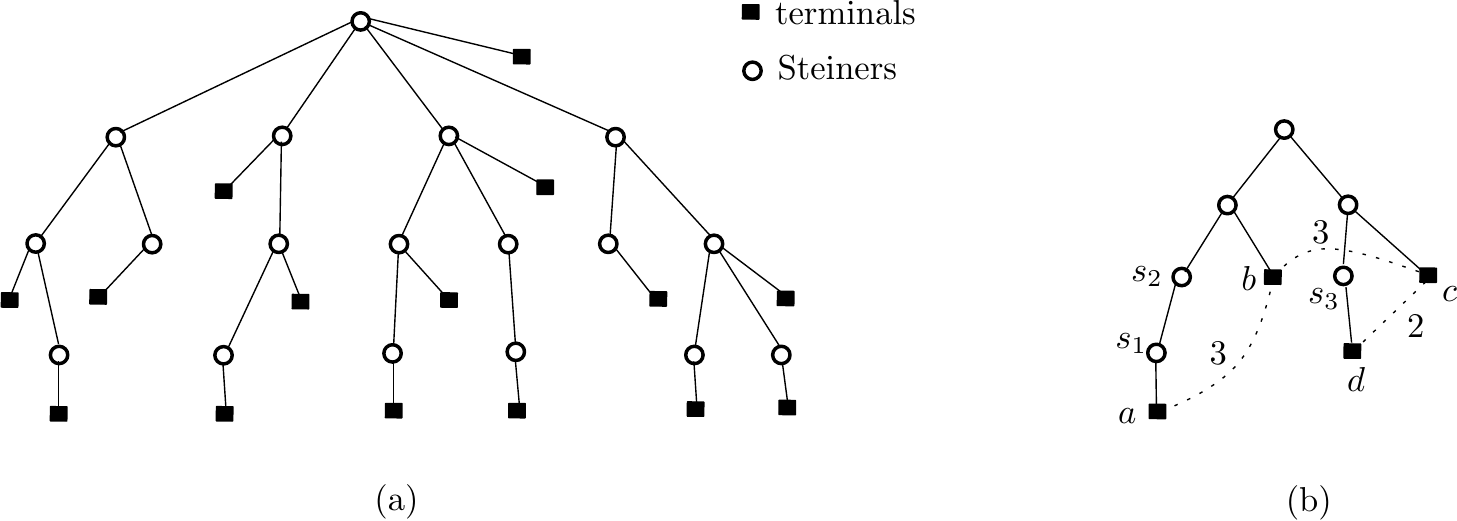}
    \caption{(a) The rooted tree, and (b) the construction of $T'_j$.}
    \label{fig:fig2}
\end{figure}

In the example in Figure~\ref{fig:fig2}(b), we first select the terminal $a$, which is the deepest one, we connect it to the terminal $b$ by an edge of weight 3 and we remove the points $s_1$ and $s_2$. Next, we select the terminal $d$, we connect it to the terminal $c$ by an edge of weight 2 and we remove the point $s_3$. In the last iteration, we select the terminal $b$, we connect it to the terminal $c$ by an edge of weight 3 and we remove all of the remaining points.

Notice that, since, in each iteration, we select the deepest terminal, we add an edge $(p,q)$, of weight $w(p,q)$, and we remove at least $\left\lfloor w(p,q)/2 \right\rfloor$ Steiner points from $T_j^*$. This implies that $C(T'_j)=\sum_{e\in T'_j}\left\lfloor w(e)/2 \right\rfloor \le k_j$, where $k_j$ is the number of Steiner points in $T_j^*$. Moreover, the union $T'$ of the trees $T'_j$ is a spanning tree of $G'$ and has $C(T') \le k$. Thus, since $T$ is a minimum spanning tree of $G'$, we have $C(T) \le C(T') \le k$.
\end {proof}

We now describe our approximation algorithm. We traverse the edges of $E$ in the sorted order and, for each edge $e_i\in E$, we construct a minimum spanning tree $T$ of $G_P=(P,E_P)$ and check whether $C(T) \le k$. If so, we construct a $k$-ST of $P$, otherwise, we move to the next edge $e_{i+1}$.

\floatname{algorithm}{Algorithm}

\begin{algorithm}[htp]
\caption{$EBST$($G=(V,E),P,k$)}\label{proc:proc1}
\begin{algorithmic}[1]

\STATE $C(T)\leftarrow\infty$
\STATE $G_P=(P,E_P)\leftarrow$ the complete graph over $P$
\STATE $i\leftarrow 0$
\WHILE {$C(T)>k$}
\STATE $i\leftarrow i+1$
\STATE construct the graph $G_i$
\FOR {each edge $(p,q)\in E_P$}
		\STATE $w(p,q)\leftarrow$ the number of Steiner points in $\delta_i(p,q)$
\ENDFOR
\STATE construct a minimum spanning tree $T$ of $G_P$ under $w$
\STATE $C(T)\leftarrow\sum_{e\in T}\left\lfloor w(e)/2 \right\rfloor$
\ENDWHILE
\STATE $Construct$-$k$-$ST$($T,G_i$) 

\end{algorithmic}
\end{algorithm}  

The construction of a $k$-ST is done as follows. For each edge $e=(p,q)\in T$, we select $\left\lfloor w(e)/2 \right\rfloor$ Steiner points on any shortest Steiner path between $p$ and $q$ in $G_i$, such that, the path from $p$ to $q$ that passes through these points has a bottleneck at most $2|e_i|$, and we connect these points to form a path; see Figure~\ref{fig:fig3}. Clearly, the obtained Steiner tree contains at most $k$ Steiner points and its bottleneck is at most $2|e_i|$.

\begin{figure}[htp]
    \centering
        \includegraphics[width=.65\textwidth, height=0.3\textwidth]{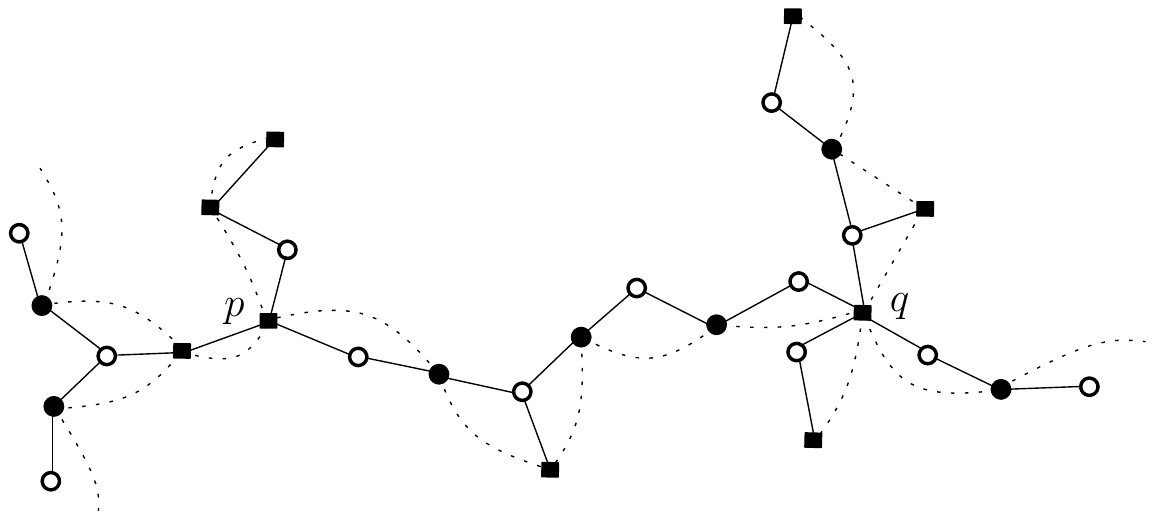}
    \caption{The constructed $k$-ST consists of the black circles and the dotted lines.}
    \label{fig:fig3}
\end{figure}

\begin{lemma} \label{lemma:lemma3.4}
The algorithm above constructs a $k$-ST of $P$ with bottleneck at most 2 times the bottleneck of an optimal $k$-ST.
\end{lemma} 

\begin{proof}[\bf{\textit{Proof:}}]
Let $e_i$ be the first edge satisfying the condition $C(T) \le k$. Thus, by Lemma~\ref{lemma:lemma3.3}, the bottleneck of any $k$-ST in $G$ is at least $|e_i|$, and, therefore, the constructed $k$-ST has a bottleneck at most 2 times the bottleneck of an optimal $k$-ST.
\end{proof}

\begin{lemma} \label{lemma:lemma3.5}
The algorithm above has a polynomial running time.
\end{lemma}

\begin{proof}[\bf{\textit{Proof:}}]
$G_i$ can be constructed in $\O((n+m)^2)$ time. In order to construct the graph $G_P$, we can compute in $\O((n+m)^3)$ time the shortest Steiner paths between each pair of points in $P$~\cite{Cormen01}. Once $G_P$ is constructed, computing a minimum spanning tree of $G_P$ can be done in $\O(n^2)$ time, and selecting the relevant Steiner points can be done in $\O(k(n+m))$ time.
\end{proof}

By combining Lemma~\ref{lemma:lemma3.4} and Lemma~\ref{lemma:lemma3.5}, we get the following theorem.
\begin{theorem}
There exists a polynomial-time approximation algorithm with performance ratio 2 for the $k$-BST problem.
\end{theorem}


\section{Conclusion}\label{sec:Sec4}
In this paper, we studied the problem of finding bottleneck Steiner trees in the Euclidean plane. We  proved that the $k$-BST problem in the plane does not admit any approximation algorithm with performance ratio less than $\sqrt{2}$, unless $P=NP$, and that there exists a polynomial-time approximation algorithm with performance ratio 2. It would be interesting to find better approximation algorithm for the $k$-BST problem. Another interesting question is how efficient can one solve the $k$-BST problem for a constant $k>0$?



\bibliographystyle{plain}
\bibliography{ref}

\end{document}